\documentclass[a4paper,UKenglish]{lipics-v2018}

\usepackage{microtype}

\graphicspath{{./img/}}

\usepackage{booktabs}
\newcommand{\Lang}[1]{%
  \ifmmode{%
    \text{\normalfont\textsc{#1}}%
  }%
  \else
  {\normalfont\textsc{#1}}%
  \fi}

\newcommand\Class[1]{%
    \mathchoice%
    {\text{\normalfont\fontsize{9pt}{10pt}\selectfont$\mathrm{#1}$}}%
    {\text{\normalfont\fontsize{9pt}{10pt}\selectfont$\mathrm{#1}$}}%
    {\text{\normalfont$\mathrm{#1}$}}%
    {\text{\normalfont$\mathrm{#1}$}}%
  }

\usepackage{ifthen}
\usepackage{framed}

\DeclareSymbolFont{extraup}{U}{zavm}{m}{n}
\DeclareMathSymbol{\varheart}{\mathalpha}{extraup}{86}
\DeclareMathSymbol{\vardiamond}{\mathalpha}{extraup}{87}

\def\tw{\mathrm{tw}}
\def\poly{\mathrm{poly}}
\def\leaf{\text{leaf}}
\def\introduce{\text{introduce}}
\def\forget{\text{forget}}
\def\join{\text{join}}
\def\edge{\text{edge}}

\def\arity{\mathrm{ar}}
\def\bit{\mathrm{bit}}
\def\symmetric{\mathrm{symmetric}}
\def\asymmetric{\mathrm{asymmetric}}

\def\jdrasil{\textsf{Jdrasil}}
\def\jatatosk{\textsf{Jatatosk}}
\def\dflat{\textsf{D-Flat}}
\def\sequoia{\textsf{Sequoia}}

\def\code#1{\texttt{#1}}

\theoremstyle{plain}
\newtheorem{observation}[theorem]{Observation}

\usepackage{xspace}
\newcommand{\ie}{i.\,e.\xspace}
\newcommand{\eg}{e.\,g.\xspace}

\usepackage{tikz}

\definecolor{uni}{HTML}{006374}
\definecolor{uniorange}{RGB}{236,116,4}
\definecolor{unired}{RGB}{228,32,50}
\definecolor{uniyellow}{RGB}{250,187,0}
\definecolor{unigreen}{RGB}{149,188,14}
\definecolor{uniblue}{RGB}{0,106,163}

\colorlet{col1}{uniorange}
\colorlet{col2}{uniblue}
\colorlet{col3}{unigreen}
\colorlet{col4}{unired}

\newcounter{todos}
\newcommand{\todo}[2][]{
  \stepcounter{todos}
  \def\bannachempty{#1}\ifx\bannachempty\empty%
   \textcolor{red!70!black}{(\arabic{todos})}%
   \marginpar{%
      \leavevmode{\raggedright\textcolor{red!70!black}{(\arabic{todos}) #2}}
    }%
  \else%
    \leavevmode{\raggedright\textcolor{red!70!black}{(\arabic{todos}) #2}}
  \fi  
}

\newcommand{\logicalObject}[7]{%
  \framebox{\parbox[t]{6.7cm}{%
    \underline{#1}\hfill\def\bannachempty{#2}\ifx\bannachempty\empty\else\cite{#2}\fi\vspace{.5ex}\\    
  \begin{tabular}{lp{4cm}}
    \#Bit:     & {\raggedright #3}\\
    Introduce: & {\raggedright #4} \\
    Forget:    & {\raggedright #5} \\
    Edge:      & {\raggedright #6}\\
    Join:      & {\raggedright #7} \\    
  \end{tabular}%
  }}
}

\usepackage{newfloat}
\DeclareFloatingEnvironment[name=Experiment,within=none]{experiment}

\title{Practical Access to Dynamic Programming on Tree Decompositions}

\author{Max Bannach}{Institute for Theoretical Computer Science,
  Universit\"at zu L\"ubeck, L\"ubeck, Germany}{bannach@tcs.uni-luebeck.de}{https://orcid.org/0000-0002-6475-5512}{}

\author{Sebastian Berndt}{Department of Computer Science,
Kiel University, Kiel, Germany}{seb@informatik.uni-kiel.de}{https://orcid.org/0000-0003-4177-8081}{}

\authorrunning{M.~Bannach and S.~Berndt}

\Copyright{Max Bannach and Sebastian Berndt}

\subjclass{\ccsdesc[100]{Theory of computation~Design and analysis of algorithms}}

\keywords{fixed-parameter tractability, treewidth, model-checking}

\category{}

\relatedversion{}

\supplement{}

\funding{}

\EventEditors{John Q. Open and Joan R. Access}
\EventNoEds{2}
\EventLongTitle{42nd Conference on Very Important Topics (CVIT 2016)}
\EventShortTitle{CVIT 2016}
\EventAcronym{CVIT}
\EventYear{2016}
\EventDate{December 24--27, 2016}
\EventLocation{Little Whinging, United Kingdom}
\EventLogo{}
\SeriesVolume{42}
\ArticleNo{23}
\nolinenumbers 
\hideLIPIcs  

\begin{document}

\maketitle

\begin{abstract}
  Parameterized complexity theory has lead to a wide range of algorithmic
  breakthroughs within  the last decades, but the practicability of these
  methods for real-world problems is still not well understood. We investigate the practicability of one
  of the fundamental approaches of this field: dynamic programming on tree
  decompositions. Indisputably, this is a key technique in parameterized
  algorithms and modern algorithm design. Despite the enormous impact of this
  approach in theory, it still has very little influence on practical
  implementations. The reasons for this phenomenon are manifold. One of them is
  the simple fact that such an implementation requires a long chain of
  non-trivial tasks (as computing the decomposition, preparing
  it,\dots). We provide an easy way to implement such dynamic programs that only
  requires the definition of the update rules.
  With this interface, dynamic programs for
  various problems, such as \Lang{3-coloring}, can be implemented easily in
  about 100 lines of structured Java code.

  The theoretical foundation of the success of dynamic programming on tree
  decompositions is well understood due to Courcelle's celebrated theorem, which
  states that every \textsc{MSO}-definable problem can be efficiently solved if
  a tree decomposition of small width is given. We seek to provide practical
  access to this theorem as well, by presenting a lightweight model-checker for
  a small fragment of \textsc{MSO}. This fragment is powerful enough to
  describe many natural problems, and our model-checker turns out to be very
  competitive against similar state-of-the-art tools.
\end{abstract}

\section{Introduction}
Parameterized algorithms aim to solve intractable
problems on instances where some parameter tied to the complexity of the
instance is small. This line of research has seen enormous growth in the last
decades and produced a wide range of algorithms~\cite{CyganFKLMPPS15}. More formally, a
problem is \emph{fixed-parameter tractable} (in \textsc{fpt}), if every instance
$I$ can be solved in time $f(\kappa(I))\cdot \poly(|I|)$ for a
computable function $f$, where
$\kappa(I)$ is the \emph{parameter} of $I$. While the impact of parameterized
complexity to the theory of algorithms and complexity cannot be overstated, its
practical component is much less understood. Very recently, the investigation of
the practicability of fixed-parameter tractable algorithms for real-world
problems has started to become an important subfield (see \eg
\cite{pace,fellows2018fpt}). We investigate the practicability of dynamic
programming on tree decompositions~--~one of the most fundamental techniques of
parameterized algorithms.
%
A general result explaining the usefulness
of tree decompositions was given by Courcelle in \cite{courcelle1990monadic},
who showed that \emph{every} property that can be expressed in monadic
second-order logic is fixed-parameter tractable if it is parameterized by tree
width. By combining this result (known as Courcelle's Theorem) with the
$f(\tw(G))\cdot |G|$ algorithm of Bodlaender \cite{bodlaender1996linear} to
compute an optimal tree decomposition in \textsc{fpt}-time, a wide range of
graph-theoretic problems is known to be solvable on these tree-like graphs.
Unfortunately, both ingredients of this approach are very expensive in practice.

One of the major achievements concerning practical parameterized algorithms was
the discovery of a practically fast algorithm for treewidth due to Tamaki
\cite{tamaki2017pid}. Concerning Courcelle's Theorem, there are currently two
contenders concerning efficient implementations of it: \dflat, an Answer Set
Programming (\textsc{ASP}) solver for problems on tree
decompositions~\cite{abseher2014d}; and \sequoia, an \textsc{MSO} solver based
on model checking games~\cite{langer2013fast}. Both solvers allow to
solve very general problems and the corresponding overhead might, thus, be large
compared to a straightforward implementation of the dynamic programs for
specific problems.

\subparagraph*{Our Contributions}
In order to study the practicability of dynamic programs on tree decompositions,
we expand our tree decomposition library \jdrasil\ with an easy to use interface
for such programs: The user only needs to specify the \emph{update
  rules} for the different kind of nodes within the tree decomposition. The
remaining work~--~computing a suitable optimized tree decomposition and
performing the actual run of the dynamic program~--~are done by \jdrasil. This
allows users to implement a wide range of algorithms within very few lines of
code and, thus, gives the opportunity to test the practicability of these
algorithms quickly. This interface is presented in Section~\ref{sec:interface}.

While \dflat\ and \sequoia\ solve very general problems, the experimental
results of Section~\ref{section:Experiments} show that  na\"{\i}ve
implementations of dynamic programs might be much more efficient. In order to
balance the generality of \textsc{MSO} solvers and the speed of direct
implementations, we introduce a small \textsc{MSO} fragment, that avoids
quantifier alternation, in
Section~\ref{section:msoSolver}. By concentrating on this fragment, we are able
to build a model-checker, called \jatatosk, that runs nearly as fast as direct 
implementations of the dynamic programs.  To show the feasibility of our
approach, we compare the running times of \dflat, \sequoia, and \jatatosk\ for
various problems. It turns out that \jatatosk\ is competitive against
the other solvers and, furthermore, its behaviour is much more consistent (\ie it
does not fluctuate greatly on similar instances). We conclude that
concentrating on just a small fragment of \textsc{MSO} gives rise to practically fast
solvers, which are still able to solve a large class of problems on graphs of
bounded treewidth. 

\section{Preliminaries}
All graphs considered in this paper are undirected, that
is, they consists of a set of vertices $V$ and of a symmetric edge-relation $E\subseteq V\times V$. We assume the
reader to be familiar with basic graph theoretic terminology, see for
instance~\cite{Diestel}. A \emph{tree decomposition} of a graph
$G=(V,E)$ is a tuple $(T,\iota)$ consisting of a rooted tree $T$ and a
mapping $\iota$ from nodes of $T$ to sets of vertices of $G$ (which we
call \emph{bags}) such that (1) for all $v\in V$ there is a node $n$
in $T$ with $v\in\iota(n)$, (2) for every edge $\{v,w\}\in E$ there is
a node $m$ in $T$ with $\{v,w\}\subseteq\iota(m)$, and (3) the set
$\{\,x\mid v\in\iota(x)\,\}$ is connected in $T$ for every $v\in
V$. The \emph{width} of a tree decomposition is the maximum size of
one of its bags minus one, and the \emph{treewidth} of $G$, denoted
by $\tw(G)$, is the minimum width any tree decomposition of $G$ must
have. 

In order to describe dynamic programs
over tree decompositions, it turns out be helpful to transform a tree
decomposition into a more structured one. A \emph{nice tree decomposition} is a triple
$(T,\iota,\eta)$ where $(T,\iota)$ is a tree decomposition and
$\eta\colon V(T)\rightarrow\{\leaf, \introduce, \forget, \join\}$ is a
labeling such that (1) nodes labeled ``$\leaf$'' are exactly the leaves
of $T$, and the bags of these nodes are empty; (2) nodes $n$
labeled ``$\introduce$'' or ``$\forget$'' have exactly one child $m$ such that
there is exactly one vertex $v\in V(G)$ with either
$v\not\in\iota(m)$ and $\iota(n)=\iota(m)\cup\{v\}$ or
$v\in\iota(m)$ and $\iota(n)=\iota(m)\setminus\{v\}$,
respectively; (3) nodes $n$ labeled ``$\join$'' have exactly two children
$x,y$ with $\iota(n)=\iota(x)=\iota(y)$. A \emph{very nice tree
  decomposition} is a nice tree decomposition that also has exactly one node
labeled ``$\edge$'' for every
$e\in E(G)$, which virtually introduces the edge $e$ to the bag~--~i.\,e.,
whenever we introduce a vertex, we assume it to be ``isolated'' in the
bag until its incident edges are introduced. It is well known that
any tree decomposition can efficiently be transformed into a very nice one without
increasing its width (essentially traverse through the tree and
``pull apart'' bags)~\cite{CyganFKLMPPS15}. Whenever we talk about
tree decompositions in the rest of the paper, we actually mean very
nice tree decompositions. However, we want to stress out that all our
interfaces also support ``just'' nice tree decompositions.

We assume the reader to be familiar with basic logic terminology and
give just a brief overview over the syntax and semantic of
monadic second-order logic (\textsc{MSO}), see for instance~\cite{FlumG06}
for a detailed introduction. A \emph{vocabulary} (or \emph{signature})
$\tau=(R_1^{a_1},\dots,R_n^{a_n})$ is a set of \emph{relational
  symbols} $R_i$ of arity $a_i\geq 1$. A \emph{$\tau$-structure}
is a set $U$~--~called \emph{universe}~--~together with an
\emph{interpretation} $R_i^U\subseteq R^{a_i}$ of the relational
symbols. Let $x_1,x_2,\dots$ be a sequence of \emph{first-order
  variables} and $X_1,X_2,\dots$ be a sequence of
\emph{second-order variables} $X_i$ of arity $\arity(X_i)$. The atomic
$\tau$-formulas are $x_i=x_j$ for two first-order variables and
$R(x_{i_{1}},\dots, x_{i_{k}})$, where $R$ is either a
relational symbol or a second-order variable of arity $k$. The set of
$\tau$-formulas is inductively defined by (1) the set of atomic
$\tau$-formulas; (2) Boolean connections $\neg \phi$,
$(\phi\vee\psi)$, and $(\phi\wedge\psi)$ of $\tau$-formulas $\phi$
and $\psi$; (3) quantified formulas $\exists x\phi$ and $\forall
x\phi$ for a first-order variable $x$ and a $\tau$-formula $\phi$; (4)
quantified formulas $\exists X\phi$ and $\forall X\phi$ for a
second-order variable $X$ of arity $1$ and a $\tau$-formula
$\phi$. The set of \emph{free variables} of a formula $\phi$ consists of the
variables that appear in $\phi$ but are not bounded by a
quantifier. We denote a formula $\phi$ with free variables
$x_1,\dots,x_{k},X_1,\dots,X_{\ell}$ as
$\phi(x_1,\dots,x_{k},X_1,\dots,X_{\ell})$. Finally, we say a
$\tau$-structure $\mathcal{S}$ with an universe $U$ is a \emph{model}
of an $\tau$-formula $\phi(x_1,\dots,x_{k},X_1,\dots,X_{\ell})$ if
there are elements $u_1,\dots,u_k\in U$ and relations
$U_1,\dots,U_{\ell}$ with $U_i\subseteq U^{\arity(X_i)}$ with
 $\phi(u_1,\dots,u_k,U_1,\dots,U_{\ell})$ being true in
$\mathcal{S}$. We write
$\mathcal{S}\models\phi(u_1,\dots,u_k,U_1,\dots,U_{\ell})$ in this case.

\begin{example}
 Graphs can be modeled as $\{E^2\}$-structures with a symmetric
 interpretation of $E$. Properties such as ``is 3-colorable''
 can then be described by formulas as:
 \[
  \tilde\phi_{\text{3col}}=\exists R\exists G\exists B\; (\forall x\,
  R(x)\vee G(x)\vee B(x))\wedge(\forall x\forall y\,
  E(x,y)\rightarrow\bigwedge_{\makebox[0pt]{\footnotesize$C\in\{R,G,B\}$}}\neg C(x)\vee\neg C(y)).
\]
 For instance, we have $\raisebox{-3.55508pt}{\includegraphics{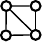}}\models\tilde\phi_{\text{3col}}$ and
 $\raisebox{-3.55508pt}{\includegraphics{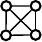}}\not\models\tilde\phi_{\text{3col}}$. 
We write $\tilde{\phi}$ whenever a more refined version of $\phi$ will be given 
later on.
\end{example}

The \emph{model-checking} problem asks, given a logical structure
$\mathcal{S}$ and a formula $\phi$, if $\mathcal{S}\models\phi$
holds. A \emph{model-checker} is a program that solves this problem and
outputs an assignment to its free and
bounded variables if $\mathcal{S}\models\phi$ holds.

\section{An Interface for Dynamic Programming on Tree  Decompositions}
\label{sec:interface}
It will be convenient to recall a classical viewpoint of dynamic programming on
tree decompositions to illustrate why our interface is designed the way
it is. We will do so by the
guiding example of $3\Lang{-coloring}$: Is it possible to color vertices of a given
graph with three colors such that adjacent vertices never share the
same color? Intuitively, a dynamic program for $3\Lang{-coloring}$
will work bottom-up on a very nice tree decomposition and manages a
set of possible colorings per node. Whenever a vertex is introduced,
the program ``guesses'' a color for this vertex; if a vertex is
forgotten we have to remove it from the bag and identify configurations
that become eventually equal; for join bags we just have to take the
configurations that are present in both children; and for edge bags we
have to reject colorings in which both endpoints of the
introduced edge have the same color. To formalize this vague
algorithmic description, we view it from the
perspective of 
automata theory.

\subsection{The Tree Automaton Perspective}\label{section:treeAutomaton}
Classically, dynamic programs on tree decompositions are described in
terms of tree automata~\cite{FlumG06}. Recall that in a very nice tree
decomposition the tree $T$ is rooted and binary; 
we assume that the children of $T$ are ordered. The mapping
$\iota$ can then be seen as a function that maps the nodes of $T$ to symbols from some
alphabet $\Sigma$.
A  na\"{\i}ve approach to manage $\iota$ would yield a huge
alphabet (depending on the size of the graph). We thus define the so called \emph{tree-index},
which is a
map $\text{idx}\colon V(G)\rightarrow \{0,\dots,\tw(G)\}$ such that
no two vertices that appear in the same bag share a common
tree-index. The existence of such an index follows directly from the
property that every vertex is forgotten exactly once: We can simply
traverse $T$ from the root to the leaves and assign a free index to a
vertex $V$ when it is forgotten, and release the used index once we reach
an introduce bag for $v$. The symbols of $\Sigma$ then only contain the
information for which tree-index there is a vertex in the bag. From a
theoreticians perspective this means that $|\Sigma|$ depends only on the
treewidth; from a programmers perspective the tree-index makes it
much easier to manage data structures that are used by the dynamic program.
\begin{definition}[Tree Automaton]
A nondeterministic bottom-up \emph{tree automaton}
is a tuple $A=(Q,\Sigma,\Delta,F)$ where $Q$ is a set of \emph{states} with
a subset $F\subseteq Q$ of \emph{accepting states}, $\Sigma$ is an \emph{alphabet},
and
$\Delta\subseteq(Q\cup\{\bot\})\times(Q\cup\{\bot\})\times\Sigma\times
Q$ is a \emph{transition relation} in which $\bot\not\in Q$ is a special
symbol to treat nodes with less than two children. The automaton is
\emph{deterministic} if for every $x,y\in Q\cup\{\bot\}$ and
every $\sigma\in\Sigma$ there is exactly one $q\in Q$ with
$(x,y,\sigma,q)\in\Delta$.
\end{definition}
\begin{definition}[Computation of a Tree Automaton]
The \emph{computation of a tree automaton} $A=(Q,\Sigma,\Delta,F)$ on
a labeled tree $(T,\iota)$ with $\iota\colon V(T)\rightarrow\Sigma$ and root
$r\in V(T)$
is an assignment $q\colon V(T)\rightarrow Q$ such that for all $n\in
V(T)$ we have (1) $(q(x),q(y),\iota(n),q(n))\in\Delta$ if $n$ has two
children $x$, $y$; (2) $(q(x),\bot,\iota(n),q(n))\in\Delta$ or
$(\bot, q(x),\iota(n),q(n))\in\Delta$ if $n$ has one
child $x$; (3) $(\bot,\bot,\iota(n),q(n))\in\Delta$ if $n$ is a
leaf. The computation is \emph{accepting} if $q(r)\in F$.
\end{definition}
\subparagraph*{Simulating Tree Automata}
A dynamic program for a decision problem can be formulated as a
nondeterministic tree automaton that works on the decomposition, see the left
side of Figure~\ref{figure:treeAutomaton} for a detailed example. Observe that a
nondeterministic tree automaton $A$ will process a labeled tree $(T,\iota)$ with
$n$ nodes in time $O(n)$. When we simulate such an automaton deterministically,
one might think that a running time of the form $O(|Q|\cdot n)$ is sufficient,
as the automaton could be in any potential subset of the $Q$ states at some node
of the tree. However, there is a pitfall: For every node we have to compute the
set of potential states of the automaton depending on the sets of potential
states of the children of that node, leading to a quadratic dependency on $|Q|$.
This can be avoided for transitions of the form $(\bot, \bot, \iota(x), p)$,
$(q, \bot, \iota(x), p)$, and $(\bot, q, \iota(x), p)$, as we can collect
potential successors of every state of the child and compute the new set of
states in linear time with respect to the cardinality of the set. However,
transitions of the form $(q_i, q_j, \iota(x), p)$ are difficult, as we now have
to merge two sets of states. In detail, let $x$ be a node with children $y$ and
$z$ and let $Q_y$ and $Q_z$ be the set of potential states in which the
automaton eventually is in at these nodes. To determine $Q_x$ we have to check
for every $q_i\in Q_y$ and every $q_j\in Q_z$ if there is a $p\in Q$ such that
$(q_i,q_j,\iota(x),p)$. Note that the number of states $|Q|$ can be quite large
even for moderately sized parameters $k$, as $|Q|$ is typically of size
$2^{\Omega(k)}$, and we will thus try to avoid
this quadratic blow-up. 
\begin{observation}\label{obs:treeAutomaton}
A tree automaton can be simulated in time
$O(|Q|^2\cdot n)$.
\end{observation}
Unfortunately, the quadratic factor in the simulation cannot be
avoided in general, as the automaton may very well contain a
transition for all possible pairs of states. However, there are some
special cases in which we can circumnavigate the increase in the
running time.
\begin{definition}[Symmetric Tree Automaton]
  A \emph{symmetric} nondeterministic bottom-up tree automaton is a
  nondeterministic bottom-up tree automaton $A=(Q,\Sigma,\Delta,F)$ in
  which all transitions $(l,r,\sigma,q)\in\Delta$ satisfy either
  $l=\bot$, $r=\bot$, or $l=r$.
\end{definition}
Assume as before that we wish to compute the set of potential states
for a node $x$ with children $y$ and $z$. Observe that in a symmetric
tree automaton it is sufficient to consider the set $Q_y\cap Q_z$ and that the
intersection of two sets can be computed in linear time if we take
some care in the design of the underlying data structures.
\begin{observation}\label{obs:symmetricTreeAutomaton}
A symmetric tree automaton can be simulated in time
$O(|Q|\cdot n)$.
\end{observation}
The right side of
Figure~\ref{figure:treeAutomaton} illustrates the
deterministic simulation of a symmetric tree automaton. The massive
time difference in the simulation of tree automata and
symmetric tree automata significantly influenced the design of the
algorithms in Section~\ref{section:msoSolver}, in which we try to
construct an automaton that is 1) ``as symmetric as possible'' and 2)
allows to take advantage of the ``symmetric parts'' even if the automaton
is not completely symmetric. 

\begin{figure}[h]
  \centering%
  \includegraphics{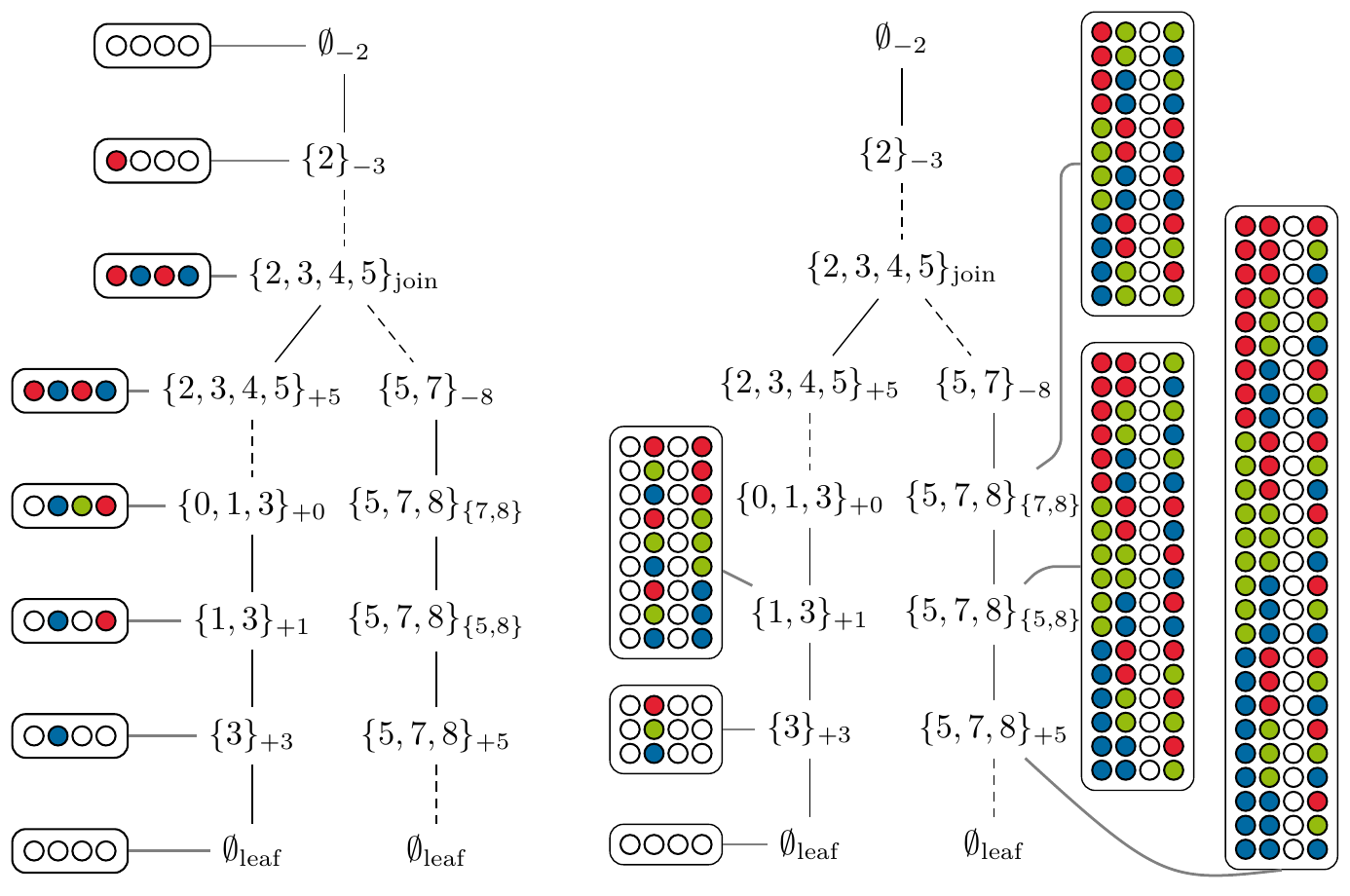}%
  \caption{The \emph{left} picture shows a part of a tree
    decomposition of the grid graph
    \raisebox{-3.55508pt}{\includegraphics{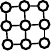}} with vertices $\{0,\dots,9\}$. The index of a bag shows the
    type of the bag: a positive sign means ``introduce'', a
    negative one ``forget'', a pair represents an ``edge''-bag, and
    text is self explanatory. Solid lines represent real edges of the
    decomposition, while dashed lines illustrate a path (i.\,e., there
    are some bags skipped). On the left branch of the decomposition a
    run of a nondeterministic tree automaton with tree-index \scalebox{.6}{$\begin{pmatrix}
      0&1&2&3&4&5&6&7&8\\
      2&3&0&1&2&3&0&1&0
    \end{pmatrix}$} for $3\Lang{-coloring}$ is illustrated. To
  increase readability, states
  of the automaton are connected to the corresponding bags with gray
  lines, and for some nodes the states are omitted. In the \emph{right} picture the same automaton is
  simulated deterministically.
  }
  \label{figure:treeAutomaton}
\end{figure}

\subsection{The Interface}
We  introduce a simple Java-interface to our library \jdrasil,
which originally was developed for the computation of tree
decompositions only. The interface is build up from two classes:
\code{StateVectorFactory} and \code{StateVector}. The only job of the
factory is to generate \code{StateVector} objects for the leaves of
the tree decomposition, or with the terms of the previous section: ``to
define the initial states of the tree automaton''. The \code{StateVector}
class is meant to model a vector of potential states in which the
nondeterministic tree automaton is at a specific node of the tree
decomposition. 
Our interface does not define at all
what a ``state'' is, or how a collection of states is managed (although most of
the times, it will be a set). The only
thing the interface requests a user to implement is the behaviour of
the tree automaton when it reaches a node of the tree-decomposition,
i.\,e., given a \code{StateVector} (for some unknown node in the tree
decomposition) and the information that the automaton reaches a
certain node, how does the \code{StateVector} for this
node look like? To this end, the interface contains the methods shown
in Listing~\ref{listing:interfaceA}.
\begin{lstlisting}[language=Java, label={listing:interfaceA},
  caption={The four methods of the interface  describe the
    behaviour of the tree automaton. Here ``T'' is a generic type for
  vertices. Each function obtains as parameter the
  current bag and a tree-index ``idx''. Other parameters
  correspond to bag-type specifics, \eg the introduced or
  forgotten vertex $v$.
  % that was introduced or forgotten.
}]
StateVector<T> introduce(Bag<T> b, T v, Map<T, Integer> idx);
StateVector<T> forget(Bag<T> b, T v, Map<T, Integer> idx);
StateVector<T> join(Bag<T> b, StateVector<T> o, Map<T, Integer> idx);
StateVector<T> edge(Bag<T> b, T v, T w, Map<T, Integer> idx);
\end{lstlisting}
This already rounds up the description of the interface, everything
else is done by \jdrasil. In detail, given a graph and an
implementation of the interface, \jdrasil\ will compute a tree
decomposition\footnote{See \cite{BannachBE17} for the concrete algorithms used
  by \jdrasil.},
transform this decomposition into a very nice
tree decomposition, potentially optimize the tree decomposition for
the following dynamic program, and finally traverse through the tree
decomposition and simulate the tree automaton described by the
implementation of the interface. The result of this procedure is the
\code{StateVector} object assigned to the root of the tree
decomposition.

\subsection{Example: 3-Coloring}\label{section:coloringSolver}
Let us illustrate the usage of the interface with our running example
of 3\Lang{-coloring}. A \code{State} of the automaton can 
be modeled as a simple integer array that stores a color (an integer)
for every vertex in the bag. A \code{StateVector} stores a set of
\code{State} objects, i.\,e., essentially a set of integer
arrays. Introducing a vertex $v$ to a \code{StateVector} therefore
means that three duplicates of each stored state have to be created,
and for every duplicate a different color has to be assigned to
$v$. Listing~\ref{listing:interfaceB} illustrates how this operation
could be realized in Java.
\begin{lstlisting}[language=Java, label={listing:interfaceB},
  caption={Exemplary implementation of the \code{introduce} method for 3\Lang{-coloring}.}]
StateVector<T> introduce(Bag<T> b, T v, Map<T, Integer> idx) {
  Set<State> newStates = new HashSet<>();
  for (State state : states) { // 'states' is the set of states
    for (int color = 1; color <= 3; color++) {
      State newState = new State(state); // copy the state
      newState.colors[idx.get(v)] = color;
      newStates.add(newState);
    }
  }
  states = newStates;
  return this;
}
\end{lstlisting}
The three other methods can be implemented in a very similar fashion:
in the \code{forget}-method we set the color of $v$ to $0$; in the \code{edge}-method we
remove states in which both endpoints of the edge have the same color;
and in the \code{join}-method we compute the intersection of the state sets of
both \code{StateVector} objects. Note that when we forget a vertex $v$,
 multiple states may become identical, which is 
handled here by the implementation of the Java \code{Set}-class, which takes care
of duplicates automatically.

A reference implementation of this 3\Lang{-coloring} solver is
publicly available~\cite{Bannach182}, and a detailed description of it can be
found in the manual of \jdrasil~\cite{BannachBE17J}. Note that this
implementation is only meant to illustrate the interface and
that we did not make any effort to optimize it. Nevertheless, this
very simple implementation (the part of the program that is
responsible for the dynamic program only contains about 120 lines of
structured Java-code) performs surprisingly well, as the
experiments in Section~\ref{section:Experiments} indicate.

\section{A Lightweight Model-Checker for a Small MSO-Fragment}\label{section:msoSolver}
Experiments with the coloring solver of the previous section have shown a huge difference in the
performance of general solvers as \dflat\ and \sequoia\ against a concrete
implementation of a tree automaton for a specific problem (see
Section~\ref{section:Experiments}). This is
not necessarily surprising, as a general solver needs to keep track of
way more information. In fact, a \textsc{MSO}-model-checker can probably (unless
$\Class{P}=\Class{NP}$) not 
run in time $f(|\phi|+\tw)\cdot\mathrm{poly}(n)$ for any elementary
function $f$~\cite{frick2004complexity}. On the
other hand, it is not clear (in general) what the concrete running time of such a solver is
for a concrete formula or problem (see e.\,g.~\cite{KNEIS2011568} for a
sophisticated analysis of some running times in \sequoia). We seek to close this gap between
(slow) general solvers and (fast) concrete algorithms. Our approach
is to concentrate only on a small fragment of \textsc{MSO},
which is powerful enough to express many natural problems, but which
is restricted enough to allow model-checking in time that matches or
is close to the running time of a concrete algorithm for the problem. As a
bonus, we will be able to derive upper bounds on the running time of the model-checker
directly from the syntax of the input formula.

Based on the interface of \jdrasil, we have implemented a publicly
available prototype called \jatatosk~\cite{Bannach18}. In Section~\ref{section:Experiments}, we perform various
experiments on different problems on multiple sets of graphs. It
turns out that  \jatatosk\ is competitive against the state-of-the-art solvers \dflat\ and
\sequoia. Arguably these two programs solve a more general problem and
a direct comparison is not entirely fair.
However, the experiments do reveal that it seems very promising to focus
on smaller fragments of \textsc{MSO} (or perhaps any other description
language) in the design of treewidth based solvers.

\subsection{Description of the Fragment}
We only consider vocabularies $\tau$ that contain the binary
relation $E^2$, and we only consider $\tau$-structures with a
symmetric interpretation of $E^2$, i.\,e., we only consider structures
that contain an undirected graph (but may also contain further relations).
The fragment of $\textsc{MSO}$ that we consider is constituted by formulas of the
form $\phi = \exists X_1\dots\exists X_k\bigwedge_{i=1}^n\psi_i$,
where the $X_j$ are second-order variables and the $\psi_i$
are first-order formulas of the form
\begin{align*}
  \psi_i\in \{\,
    &\forall x\forall y\; E(x,y)\rightarrow \chi_i,\;
    \forall x\exists y\; E(x,y)\wedge \chi_i,\;
    \exists x\forall y\; E(x,y)\rightarrow \chi_i,\\
    &\exists x\exists y\; E(x,y)\wedge \chi_i,\;
    \forall x\; \chi_i,\;
    \exists x\; \chi_i\,\}.
\end{align*}
  
Here, the $\chi_i$ are quantifier-free first-order formulas in
canonical normal form. It is easy to see
that this fragment is already powerful enough to encode many classical
problems as $3\Lang{-coloring}$ ($\tilde\phi_{\text{3col}}$ from the
introduction is part of the fragment),
or \Lang{vertex-cover} (we will discuss how to handle optimization in Section~\ref{section:optimization}):
$\tilde\phi_{\text{vc}}=\exists S\forall x\forall y\; E(x,y)\rightarrow
  S(x)\vee S(y)$.

\subsection{A Syntactic Extension of the Fragment}
Many interesting properties, such as connectivity, can easily be
expressed in \textsc{MSO}, but not directly in the fragment that we
study. Nevertheless, a lot of these properties can directly be checked by a model-checker if
it ``knows'' what kind of properties it actually checks. We present a
\emph{syntactic extension} of our
\textsc{MSO}-fragment which captures such properties. The extension
consist of three new second order quantifiers that can be used instead of
$\exists X_i$.

The first extension is a \emph{partition quantifier}, which quantifies
over partitions of the universe:
\[
  \exists^{\mathrm{partition}}X_1,\dots,X_k\equiv\exists X_1\exists
  X_2\dots\exists X_k\big(\forall x\;\bigvee_{i=1}^k X_i(x)\big)\wedge
  \big(\forall x\;\bigwedge_{i=1}^k\bigwedge_{j\neq i}\neg
  X_i(x)\wedge\neg X_j(x)\big).
\]
This quantifier has two advantages. First, formulas like
$\tilde\phi_{\text{3col}}$ can be simplified to
\[
  \phi_{\text{3col}}=\exists^{\text{partition}} R,G,B\; \forall x\forall y\,
  E(x,y)\rightarrow\bigwedge_{\makebox[0pt]{\footnotesize$C\in\{R,G,B\}$}}\neg C(x)\vee\neg C(y),
\]
and second, the model-checking problem for them can be solved more efficiently: the solver directly
``knows'' that a vertex must be added to exactly one of the sets.

We further introduce two quantifiers that work with respect to the
symmetric relation $E^2$ (recall that we only consider structures that
contain such a relation). The
$\exists^{\mathrm{connected}} X$ quantifier guesses an $X\subseteq U$
that is connected with respect to $E$ (in graph theoretic terms),
i.\,e., it quantifies over connected subgraphs. The
$\exists^{\mathrm{forest}} F$ quantifier guesses a $F\subseteq U$ that
is acyclic with respect to $E$ (again in graph theoretic terms), i.\,e., it
quantifies over subgraphs that are forests. These quantifiers are
quite powerful and allow, for instance, to express that the 
graph induced by $E^2$ contains a triangle as minor:
\begin{align*}
  \phi_{\text{triangle-minor}} =&
  \exists^{\mathrm{connected}}R\,\exists^{\mathrm{connected}}G\,\exists^{\mathrm{connected}}B\;\centerdot\\
  &\big(\forall x\, (\neg R(x)\vee\neg G(x))\wedge(\neg G(x)\vee\neg B(x))\wedge(\neg B(x)\vee\neg R(x))\big)\\
  \wedge\,& \big(\exists x\exists y\; E(x,y)\wedge R(x)\wedge G(y)\big) 
   \wedge \big(\exists x\exists y\; E(x,y)\wedge G(x)\wedge B(y)\big)\\ 
   \wedge\,& \big(\exists x\exists y\; E(x,y)\wedge B(x)\wedge R(y)\big).
\end{align*}
We can also express problems that usually require more involved
formulas in a very natural way. For instance, the 
\Lang{feedback-vertex-set} problem can be described by the following
formula (again, optimization will be handled in Section~\ref{section:optimization}):
$\tilde\phi_{\text{fvs}}=\exists S\,\exists^{\text{forest}}F\,\forall x\; S(x)\vee F(x) $.

\subsection{Description of the Model-Checker}\label{section:discription}
We describe our model-checker in terms of a nondeterministic tree automaton
that works on a tree decomposition of the graph induced by $E^2$ 
(note that, in contrast to other approaches in the literature, we do
not work on the Gaifman graph).               
We define any state of the automaton as bit-vector, and we
stipulate that the initial state at every leaf is the zero-vector. For
any quantifier or subformula, there will be some area in the bit-vector
reserved for that quantifier or subformula and we describe how state
transitions effect these bits. The ``algorithmic idea'' behind the
implementation of these transitions is not new, and a reader familiar
with folklore dynamic programs on tree decompositions (for instance
for \Lang{vertex-cover} or \Lang{steiner-tree}) will probably
recognize them. An overview over common techniques can be found in the
             standard textbooks~\cite{CyganFKLMPPS15, FlumG06}.

\subparagraph*{The Partition Quantifier} We start with a detailed description of the partition quantifier
$\exists^{\text{partition}}X_1,\dots,X_q$ (we do not implement an
additional $\exists X$ quantifier, as we can easily state
$\exists X\equiv\exists^{\text{partition}}X,\bar{X}$): Let $k$ be the
maximum bag-size of the tree decomposition. We reserve
$k\cdot\log_2 q$ bit in the state description, where each block of
length $\log_2q$ indicates in which set $X_i$ the corresponding element of the
bag is. On an introduce-bag (\eg for $v\in U$), the
nondeterministic automaton guesses an index $i\in\{1,\dots, q\}$ and sets
the $\log_2 q$ bits that are associated with the tree-index of $v$ to
$i$. Equivalently, the corresponding bits are cleared when the
automaton reaches a forget-bag. As the partition is independent of any
edges, an edge-bag does not change any of the bits reserved for the
partition quantifier. Finally, on join-bags we may only join states
that are identical on the bits describing the partition 
(as otherwise the vertices of the bag would be in different
partitions)~--~meaning this transition is symmetric with respect to
these bits (in terms of Section~\ref{section:treeAutomaton}). 

\subparagraph*{The Connected Quantifier} The next quantifier we describe
is $\exists^{\text{connected}}X$ which has to overcome the difficulty
that an introduced vertex may not be connected to the rest of the bag
in the moment it got introduced, but may be connected to it when
further vertices ``arrive''. The solution to this dilemma is to manage
a partition of the bag into $k'\leq k$ connected components
$P_1,\dots,P_{k'}$, for which we reserve $k\cdot\log_2 k$ bit in the
state description. Whenever a vertex $v$ is introduced, the automaton
either guesses that it is not contained in $X$ and clears the corresponding
bits, or it guesses that $v\in X$ and assigns some $P_i$ to
$v$. Since $v$ is isolated in the bag in the moment of its introduction
(recall that we work on a very nice tree decomposition), it requires
its own component and is therefore assigned to the smallest empty partition
$P_i$. When a vertex $v$ is forgotten, there are four possible
scenarios: 1)~$v\not\in X$, then the corresponding bits are already
cleared and nothing happens; 2)~$v\in X$ and $v\in P_i$ with
$|P_i|>1$, then $v$ is just removed and the corresponding bits are
cleared; 3)~$v\in X$ and $v\in P_i$ with
$|P_i|=1$ and there are other vertices $w$ in the bag with $w\in X$,
then the automaton rejects the configuration, as $v$ is the last
vertex of $P_i$ and may not be connected to any other partition
anymore; 4)~$v\in X$ is the last vertex of the bag that is contained in
$X$, then the connected component is ``done'', the corresponding bits
are cleared and one additional bit is set to indicate that the
connected component cannot be extended anymore. When an edge
$\{u,v\}$ is introduced, components might need to be merged. Assume $u,v\in
X$, $u\in P_i$, and $v\in P_j$ with $i<j$ (otherwise, an edge-bag does not change
the state), then we
essentially perform a classical union-operation from the well-known
union-find data structure. Hence, we assign all vertices
that are assigned to $P_j$ to $P_i$. Finally, at a join-bag we may
join two states that agree locally on the vertices that are in $X$
(i.\,e., they have assigned the same vertices to some $P_i$), however,
they do not have to agree in the way the different vertices are
assigned to $P_i$ (in fact, there does not have to be an
isomorphism between these assignments). Therefore, the transition at a join-bag has to connect
the corresponding components analogous to the edge-bags~--~in terms of
Section~\ref{section:treeAutomaton} this transition is
not symmetric. The description of the remaining quantifiers and
subformulas is very similar and presented 
in Appendix~\ref{appendix:Fragment}. 

\subsection{Extending the Model-Checker to Optimization Problems}\label{section:optimization}
As the example formulas from the previous section already indicate,
performing model-checking alone will not suffice to express many
natural problems. In fact, every graph is a model of the
formula $\tilde\phi_{\text{vc}}$ if $S$ simply contains all
vertices. It is therefore a natural extension to consider an
optimization version of the model-checking problem, which is usually
formulated as follows~\cite{CyganFKLMPPS15, FlumG06}: Given a logical structure $\mathcal{S}$, a formula $\phi(X_1,\dots,X_p)$ of the
\textsc{MSO}-fragment defined in the previous section with free
unary second-order variables $X_1,\dots,X_p$, and weight functions
$\omega_1,\dots,\omega_p$ with $\omega_i\colon
U\rightarrow\mathbb{Z}$; find $S_1,\dots,S_p$ with $S_i\subseteq U$
such that $\sum_{i=1}^p\sum_{s\in S_i}\omega_i(s)$ is \emph{minimized}
under $\mathcal{S}\models\phi(S_1,\dots,S_p)$, or conclude that
$\mathcal{S}$ is not a model for $\phi$ for any assignment of the free variables.
We can now correctly express the (actually \emph{weighted})
optimization version of \Lang{vertex-cover} as follows:
\(
  \phi_{\text{vc}}(S)=\forall x\forall y\; E(x,y)\rightarrow
  \big(S(x)\vee S(y)\big).
  \)
  
\noindent Similarly we can describe the optimization version of $\Lang{dominating-set}$ if we
assume the input does not have isolated vertices (or is reflexive),
and we can also fix the formula $\tilde\phi_{\text{fvs}}$:
\begin{align*}
  \phi_{\text{ds}}(S)=\forall x\exists y\; E(x,y)\wedge
                       \big(S(x)\vee S(y)\big),\quad 
  \phi_{\text{fvs}}(S)=\exists^{\text{forest}}F\;\forall x\, \big(S(x)\vee F(x)\big).
\end{align*}
We can also \emph{maximize} the term
$\sum_{i=1}^p\sum_{s\in S_i}\omega_i(s)$ by multiplying all weights
with $-1$ and, thus, express problems such as \Lang{independent-set}:
\(
  \phi_{\text{is}}(S)=\forall x\forall y\; E(x,y)\rightarrow
  \big(\neg S(x)\vee \neg S(y)\big).
\)
The implementation of such an optimization is straightforward:
essentially there is a partition quantifier for every free variable
$X_i$ that partitions the universe into $X_i$ and $\bar X_i$.  We assign  
a current value of $\sum_{i=1}^p\sum_{s\in S_i}\omega_i(s)$ to every
state of the automaton, which is
adapted if elements are ``added'' to some of the
free variables at introduce nodes. 
Note that, since we optimize an affine function, this
does not increase the state space: even if multiple computational
paths lead to the same state with different values at some node of the tree, it is well
defined which of these values is the optimal one. Therefore, the cost
of optimization only lies in the partition quantifier, \ie, we pay
with $k$ bits in the state description of the automaton per free
variable~--~independently of the weights.

\subsection{Handling Symmetric and Non-Symmetric Joins}
In Section~\ref{section:discription} we have defined the states of
our automaton with respect to a formula, the left side of
Table~\ref{table:bitSize} gives an overview of the number of bits we
require for the different parts of the formula. Let $\bit(\phi,k)$ be
the number of bits that we have to reserve for a formula $\phi$ and
a tree decomposition of maximum bag size $k$, i.\,e., the sum over the
required bits of each part of the formula. By Observation~\ref{obs:treeAutomaton} this
implies that we can simulate the automaton (and hence, solve the
model-checking problem) in time $O^*\big((2^{\bit(\phi,k)})^2\cdot n\big)$; or by Observation~\ref{obs:symmetricTreeAutomaton} in time
$O^*\big(2^{\bit(\phi,k)}\cdot n\big)$ if the automaton is
symmetric\footnote{The notation $O^{*}$ supresses polynomial factors.}. Unfortunately, this is not always the case, in fact, only the quantifier
$\exists^{\text{partition}}X_1,\dots,X_q$, the bits needed to
optimize over free variables, as well as the formulas that do not
require any bits, yield an symmetric tree automaton. That means that the
simulation is wasteful if we
consider a mixed formula (for instance, one that contains a partition and
a connected quantifier). To overcome this issue, we partition the bits
of the state description into two parts: first the ``symmetric'' bits
of the quantifiers $\exists^{\text{partition}}X_1,\dots,X_q$ and
the bits required for optimization, and in the ``asymmetric'' ones of all
other elements of the formula. Let $\symmetric(\phi,k)$ and
$\asymmetric(\phi,k)$ be defined analogously to $\bit(\phi,k)$. We
implement the join of states as in the following lemma, allowing 
us to deduce the
running time of the model-checker for concrete formulas. The right side of
Table~\ref{table:bitSize} provides an overview for formulas
presented here.
\begin{lemma}\label{lemma:join}
Let $x$ be a node of $T$ with children $y$ and $z$, and let $Q_y$ and
$Q_z$ be sets of states in which the automaton may be at $y$ and
$z$. Then the set $Q_x$ of states in which the automaton may be at
node $x$ can be computed in time
$O^*\big(2^{\symmetric(\phi,k)+2\cdot\asymmetric(\phi,k)}\big)$.
\end{lemma}
\begin{proof}
To compute $Q_x$, we first split $Q_y$ into $B_1,\dots,B_q$ such
that all elements in one $B_i$ share the same ``symmetric bits''. This
can be done in time $|Q_y|$ using bucket-sort. Note that we
have $q\leq2^{\symmetric(\phi,k)}$ and
$|B_i|\leq2^{\asymmetric(\phi,k)}$. With the same technique we identify for every
elements $v$ in $Q_z$ its corresponding partition $B_i$. Finally, we compare $v$ with the elements in $B_i$ to
identify those for which there is a transition in the automaton. This
yields a running time of $|Q_z|\cdot \max_{i=1}^q|B_i|\leq 2^{\bit(\phi,k)}\cdot2^{\asymmetric(\phi,k)}=2^{\symmetric(\phi,k)+2\cdot\asymmetric(\phi,k)}$.
\end{proof}
\begin{table}
  \caption{The left table shows the precise number of bit we reserve in the description of a
    state of the tree automaton for different quantifier and
    formulas. The values are with respect to a tree decomposition with
  maximum bag size $k$. The right table gives an overview of example
formulas $\phi$ used here, together with values
$\symmetric(\phi,k)$ and $\asymmetric(\phi,k)$, as well as the
precise time our algorithm will require for that
particular formula.}\label{table:bitSize}
  \begin{tabular}[t]{rl}
    \toprule
    Quantifier / Formula & Number of Bit\\
    \cmidrule(lr){1-2}
    free variables $X_1,\dots,X_q$ & $q\cdot k$\\
    $\exists^{\text{partition}}X_1,\dots,X_q$ & $k\cdot\log_2q$\\\vspace{.5ex}
    $\exists^{\text{connected}}X$ & $k\cdot\log_2 k+1$\\\vspace{.5ex}
    $\exists^{\text{forest}}X$ & $k\cdot\log_2 k$\\
    $\forall x\forall y\;E(x,y)\rightarrow\chi_i$ & $0$\\
    $\forall x\exists y\;E(x,y)\wedge\chi_i$ & $k$\\
    $\exists x\forall y\;E(x,y)\rightarrow\chi_i$ & $k+1$\\
    $\exists x\exists y\;E(x,y)\wedge\chi_i$ & $1$\\
    $\forall x\;\chi_i$ & $0$\\
    $\exists x\;\chi_i$ & $1$\\
    \bottomrule
  \end{tabular}
  \quad
  \begin{tabular}[t]{lp{2.5cm}l}
    \toprule
    $\phi$ & $\symmetric(\phi,k)$\newline$\asymmetric(\phi,k)$ & Time\\
    \cmidrule(rl){1-3}
    $\phi_{\text{3col}}$ & $k\cdot\log_2(3)$\newline$0$ & $O^*(3^k)$\\
    $\phi_{\text{vc}}(S)$ & $k$\newline$0$ & $O^*(2^k)$\\
    $\phi_{\text{ds}}(S)$ & $k$\newline$k$ & $O^*(8^k)$\\
    $\phi_{\text{triangle-minor}}$ & $0$\newline$3k\cdot\log_2(k)+3$ & $O^*(k^{6k+6})$\\
    $\phi_{\text{fvs}}(S)$ & $k$\newline$k\cdot\log_2(k)$ & $O^*(2^kk^{2k})$\\
    \bottomrule
  \end{tabular}
\end{table}

\section{Applications and Experiments}\label{section:Experiments}
In order to show the feasibility of our approach, we have performed
  experiments for widely investigated graph problems: \Lang{3-coloring},
\Lang{vertex-cover}, \Lang{dominating-set}, \Lang{independent-set}, and \Lang{feedback-vertex-set}.
All experiments where performed on an Intel Core processor containing four
cores of 3{.}2 GHz each and 8 Gigabyte RAM. The machine runs Ubuntu~17.10. \jdrasil\ was used with
 Java 1.8 and both \sequoia\ and \dflat\ were compiled with
  gcc~7.2. The implementation of \jatatosk\ uses hashing to realize
  Lemma~\ref{lemma:join}, which has no constant-time worst case guarantee but
  works well in practice.
  We use a data set that was assembled from three different sources and
  that contains graphs with 18 to 956 vertices and treewidth 3 to 13. The first
source is a collection of publicly available transit graphs from GTFS-transit feeds
\cite{gtfs} that was also used for experiments in
\cite{fichte2017sat}, the second source are
real-world instances collected in \cite{abseher2015dflat}, and the
last one are the publicly available graphs used in the PACE
challenge~\cite{pace} (we selected the ones with treewidth at most
  11). For \Lang{3-coloring} the results can be found in
  Experiment~\ref{experiment:coloring}, and for the other problems in
  Appendix~\ref{appendix:experiments}.
  
  \begin{experiment}[t]
    \caption{$\Lang{3-coloring}$}\label{experiment:coloring}
  \subcaption{Average, standard deviation, and median of the time (in seconds)
    each solver
    needed to solve $\Lang{3-coloring}$ over
    all instances of the data set.
    The best values are highlighted.}
  \label{table:coloring}
  \begin{tabular}{lllll}
    \toprule
    & \dflat & \textsf{Jdrasil-Coloring} & \jatatosk & \sequoia\\
    \cmidrule(rl){1-5}
    Average Time & 478.19 & \textbf{36.52} & 42.63 & 714.73\\
    Standard Deviation & 733.90 & \textbf{77.8} & 81.82 & 866.34\\
    Median Time & \textbf{3.5} & 21 & 24.5 & 20.5 \\
    \bottomrule
  \end{tabular}
  
  \subcaption{Comparison of solvers for the \Lang{3-coloring} problem on
    the complete data set.
    \label{figure:coloring}
  }
  \centering
  \resizebox{\textwidth}{!}{
    \includegraphics{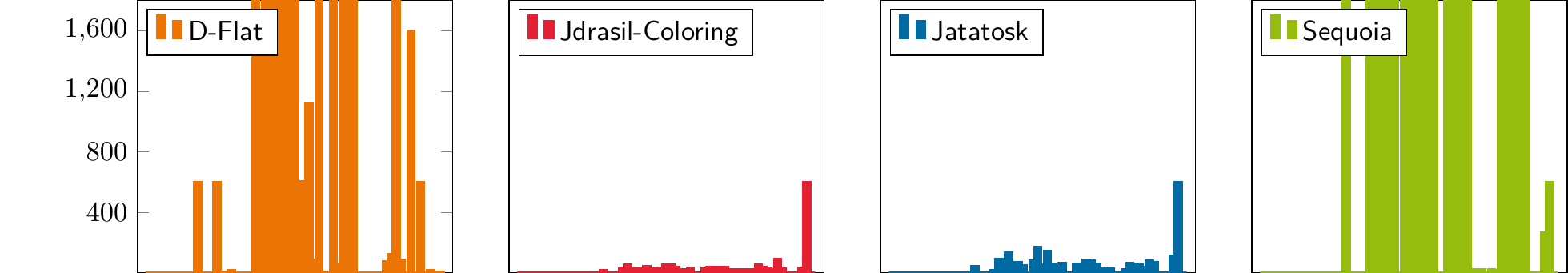}
  }
    
  \subcaption{The  left picture shows the \textcolor{col1}{difference}
    of \jatatosk\ against \dflat\ and \sequoia. A
    positive bar means that \jatatosk\ is faster by this amount in
    seconds, and a negative bar means that either \dflat\ or \sequoia\
    is faster by that amount. The bars are capped at $100$. On every instance, \jatatosk\ was compared
    against the solver that was faster on this particular instance. The
    image also shows for every instance the \textcolor{col2}{size} and
    the \textcolor{col3}{treewidth} of the input. The right image shows
    the number of instances that can be solved by each of the solvers in
    $x$ seconds, i.\,e., faster growing functions are better. The colors in this image
    are as in \textsf{(b)}.
  }
  \resizebox{\textwidth}{!}{\includegraphics{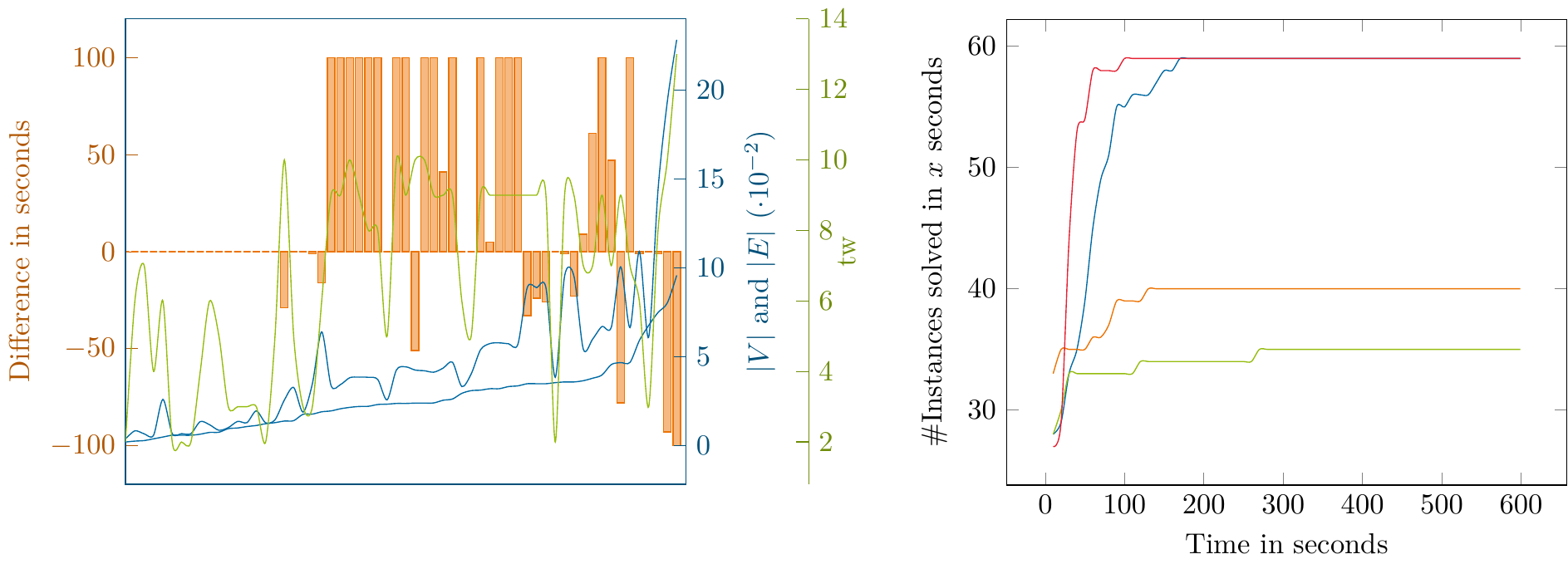}}
\end{experiment}
  \section{Conclusion and Outlook}
  We investigated the practicability of dynamic programming on tree decompositions, which 
  is arguably one of the corner stones of parameterized complexity theory.
  We implemented a simple interface for such
  programs and demonstrated how it can be used to build a competitive graph coloring solver with just a
  few lines of code. We hope that the interface allows other
  researchers to implement and explore various dynamic programs on tree decompositions.
  The whole power of such dynamic programs is well
  captured by Courcelle's Theorem, which essentially states that there is
  an efficient version of such a program for every problem definable in
  monadic second-order logic. We took a step towards practice here as
  well, by implementing a ``lightweight'' version in the form of a
  model-checker for a small fragment of the logic. By clever syntactic
  extensions, the fragment turns out to be powerful enough to express
  many natural problems such as \Lang{3-coloring},
  \Lang{feedback-vertex-set},
  and more.
%

\clearpage
\bibliography{main}

\clearpage
\appendix
\section{Technical Appendix: Description of the Fragment}\label{appendix:Fragment}
\noindent%
\begin{minipage}[t]{.49\textwidth}
\logicalObject{$\exists^{\text{forest}}X$}{}{$k\cdot\log_2k$}{%
As for $\exists^{\text{connected}}X$.%
}{%
Just clear the corresponding bits.
}{%
As for $\exists^{\text{connected}}X$, but reject if two vertices of
the same component are connected.
}{
As for $\exists^{\text{connected}}X$, but track if the join introduces
a cycle.
}
\end{minipage}
\begin{minipage}[t]{.49\textwidth}
\logicalObject{$\forall x\forall
  y\;E(x,y)\rightarrow\chi_i\vphantom{\exists^{\text{forest}}}$}{}{0}{-}{-}{Reject if $\chi_i$ is not
  satisfied for the vertices of the edge.}{-}
\end{minipage}\bigskip

\noindent
\logicalObject{$\forall x\exists y\;
  E(x,y)\wedge\chi_i$}{}{$k$}{-}{Reject if the bit corresponding to
  $v$ is not set.}{Set the bit of $v$ if $\chi_i$ is
  satisfied.}{Compute the logical-or of the bits of both states.}
\logicalObject{$\exists x\forall
  y\;E(x,y)\rightarrow\chi_i$}{}{$k+1$}{Set the corresponding
  bit.}{If the corresponding bit is set, set the additional bit.}{If
  $\chi_i$ is not satisfied, clear the corresponding bit.}{Compute the
logical-and of all but the last bit, for the last bit use a logical-or.}\bigskip

\noindent
\logicalObject{$\exists x\exists
  y\;E(x,y)\wedge\chi_i$}{}{1}{-}{-}{Set the bit if $\chi_i$ is
  satisfied.}{Compute logical-or of the bit in both states.}
\logicalObject{$\forall x\;\chi_i$ ($\exists x\;\chi_i$)}{}{0
  (1)}{Test if $\chi_i$ is satisfied and reject if not (set the bit if
so).}{-}{-}{ - (Compute logical-or of the bit in both states.)}

\clearpage
\section{Technical Appendix: Further Experiments}\label{appendix:experiments}
  We perform the experiments from Section~\ref{section:Experiments}
  for further problems. We did run every solver for a maximum of 600
  seconds on every instance for every problem. It can be seen that \jatatosk\ is competitive
  (though not superior) against its competitors, as it is faster then the
  faster of the two on many instances and its average running time is
  at most twice the time of the corresponding fastest
  algorithm. \jatatosk\ outperforms the others for \Lang{3-coloring},
  but gets outperformed for \Lang{vertex-cover} by \sequoia. The same
  holds for \Lang{independent-set}, also the difference is much
  smaller in this case.
  For \Lang{dominating-set} the situation
  is more complex, as \jatatosk\ outperforms the others on about half
  of the instances, and gets outperformed on the other
  half. Interestingly, the difference is quit high in both halves in
  both directions.
  \begin{experiment}[h]
    \caption{$\Lang{vertex-cover}$}
      \subcaption{Average, standard deviation, and median of the time (in seconds) each solver
        needed to solve $\Lang{vertex-cover}$ over
        all instances of the data set from Section~\ref{section:Experiments}. The best values are
        highlighted.}
    \label{table:vertexCover}
    \begin{tabular}{llll}
      \toprule
      & \dflat &  \jatatosk & \sequoia\\
      \cmidrule(rl){1-4}
      Average Time & 451.68 & 59.02 & \textbf{33.95}\\
      Standard Deviation  & 213.08 & 128.45 & \textbf{92.45}\\
      Median Time & 597.5 & 30 & \textbf{6}\\
      \bottomrule
    \end{tabular}
    
  \subcaption{Comparison of solvers for the \Lang{vertex-cover} problem on
    the complete data set of Section~\ref{section:Experiments}.
    The
    time is measured in seconds and the instances are sorted by the
    number of vertices, see also \textsf{(c)}.
    \label{figure:vertexCover}}
    
  \centering%
  \resizebox{\textwidth}{!}{\includegraphics{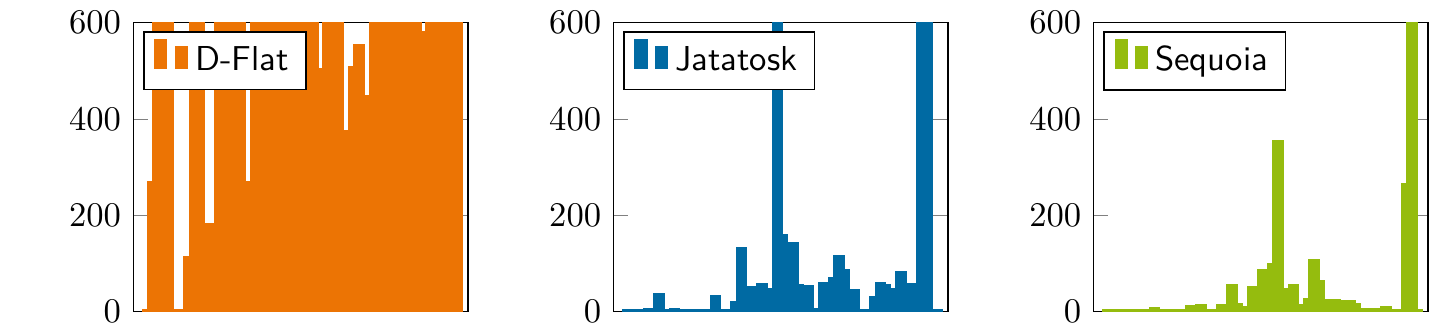}}%
 
  \subcaption{The left picture shows the \textcolor{col1}{difference}
    of \jatatosk\ against \dflat\ and \sequoia. A
    positive bar means that \jatatosk\ is faster by this amount in
    seconds, and a negative bar means that either \dflat\ or \sequoia\
    is faster by that amount. The bars are capped at $100$. On every instance, \jatatosk\ was compared
    against the solver that was faster on this particular instance. The
    image also shows for every instance the \textcolor{col2}{size} and
    the \textcolor{col3}{treewidth} of the input. The right image shows
    the number of instances that can be solved by each of the solvers in
    $x$ seconds, i.\,e., faster growing functions are better. The colors in this image
    are as in \textsf{(b)}.}
  \resizebox{\textwidth}{!}{\includegraphics{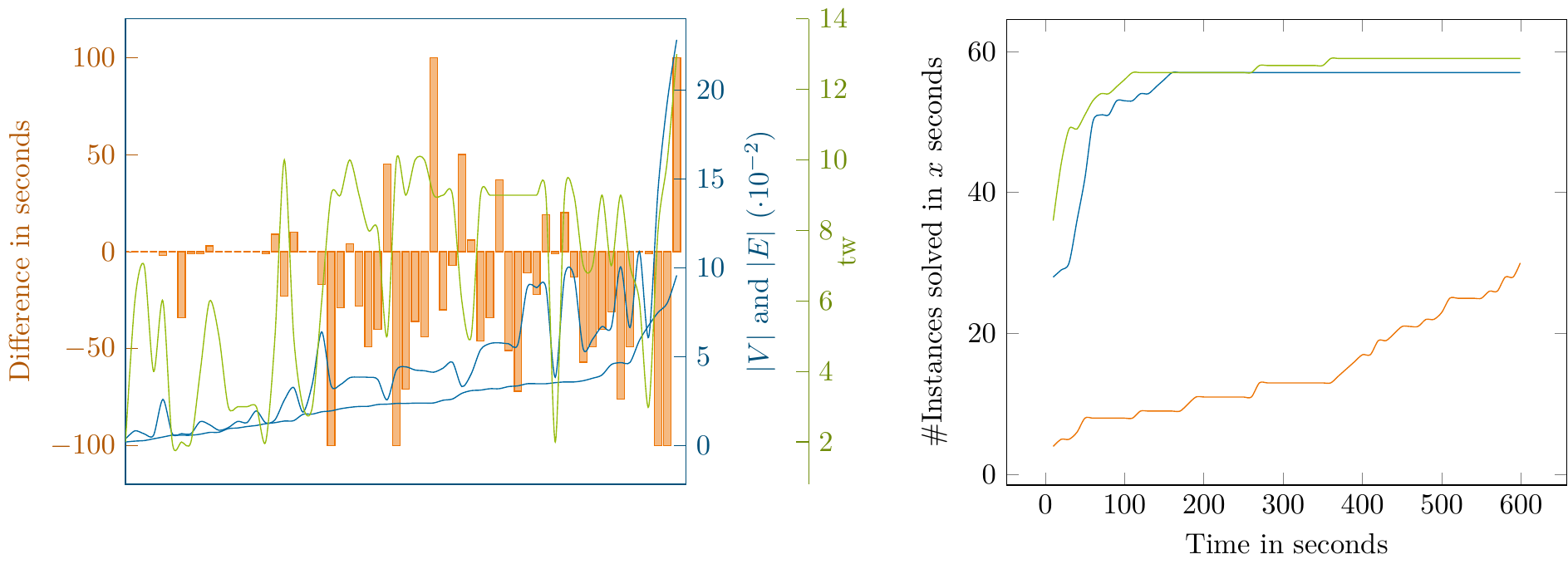}}%
\end{experiment}
  
  \begin{experiment}[h]
    \caption{$\Lang{dominating-set}$}
      \subcaption{Average, standard deviation, and median of the time (in seconds) each solver
        needed to solve $\Lang{dominating-set}$ over
        all instances of the data set from Section~\ref{section:Experiments}. The best values are
        highlighted.}
  \label{table:dominatingSet}
  \begin{tabular}{llll}
    \toprule
    & \dflat & \jatatosk & \sequoia\\
    \cmidrule(rl){1-4}
    Average Time & 420.14 &  \textbf{102.48} & 114.92\\
    Standard Deviation & 265.14 & \textbf{157.80} & 196.67\\
    Median Time & 600 & 44.5 & \textbf{20.5}\\
    \bottomrule
  \end{tabular}
  \vspace{.5cm}

  \subcaption{Comparison of solvers for the \Lang{dominating-set} problem on
    the complete data set of Section~\ref{section:Experiments}. The
    time is measured in seconds and the instances are sorted by the
    number of vertices, see also \textsf{(c)}.
    \label{figure:dominatingSet}
  }
  \centering%
  \resizebox{\textwidth}{!}{\includegraphics{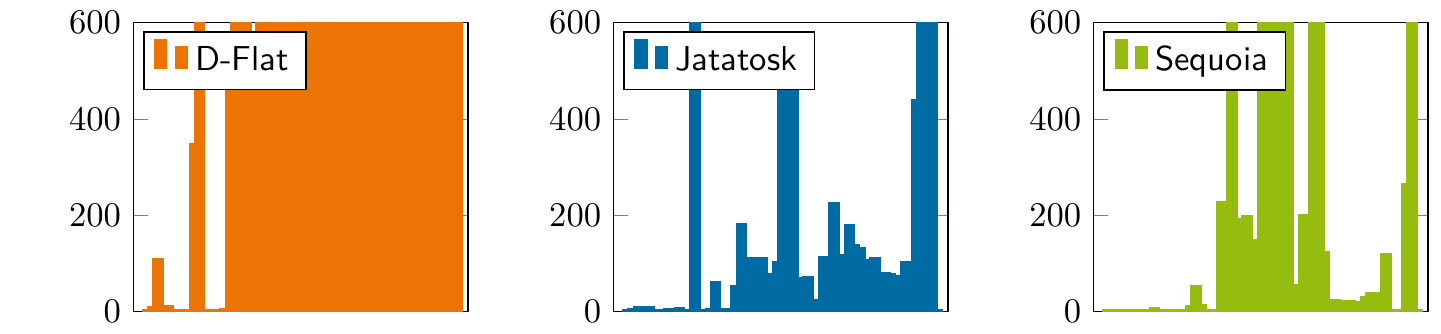}}%

  \subcaption{The left picture shows the \textcolor{col1}{difference}
    of \jatatosk\ against \dflat\ and \sequoia. A
    positive bar means that \jatatosk\ is faster by this amount in
    seconds, and a negative bar means that either \dflat\ or \sequoia\
    is faster by that amount. The bars are capped at $100$. On every instance, \jatatosk\ was compared
    against the solver that was faster on this particular instance. The
    image also shows for every instance the \textcolor{col2}{size} and
    the \textcolor{col3}{treewidth} of the input. The right image shows
    the number of instances that can be solved by each of the solvers in
    $x$ seconds, i.\,e., faster growing functions are better. The colors in this image
    are as in \textsf{(b)}.}
  \resizebox{\textwidth}{!}{\includegraphics{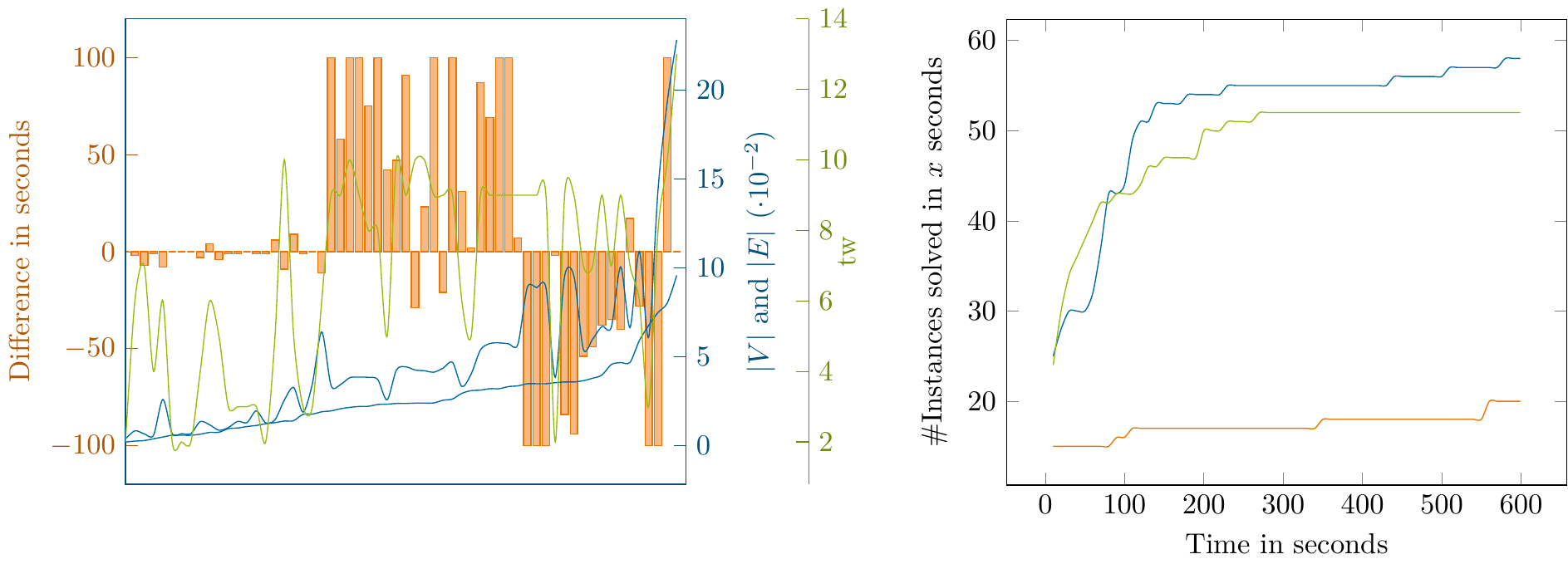}}
\end{experiment}

  \begin{experiment}[h]
    \caption{$\Lang{independent-set}$}
    \subcaption{Average, standard deviation, and median of the time (in seconds) each solver
      needed to solve $\Lang{independent-set}$ over
      all instances of the data set from Section~\ref{section:Experiments}. The best values are
      highlighted.}
  \label{table:independentSet}
  \begin{tabular}{llll}
    \toprule
    & \dflat & \jatatosk & \sequoia\\
    \cmidrule(rl){1-4}
    Average Time & 229.18 &  16.98 & \textbf{15.32}\\
    Standard Deviation & 272.64 & \textbf{18.17} & 45.53\\
    Median Time & 13 & 14 & \textbf{1}\\
    \bottomrule
  \end{tabular}
  \vspace{.5cm}

  \subcaption{Comparison of solvers for the \Lang{independent-set} problem on
    the complete data set of Section~\ref{section:Experiments}. The
    time is measured in seconds and the instances are sorted by the
    number of vertices, see also \textsf{(c)}.
    \label{figure:independentSet}
  }
  \centering%
  \resizebox{\textwidth}{!}{\includegraphics{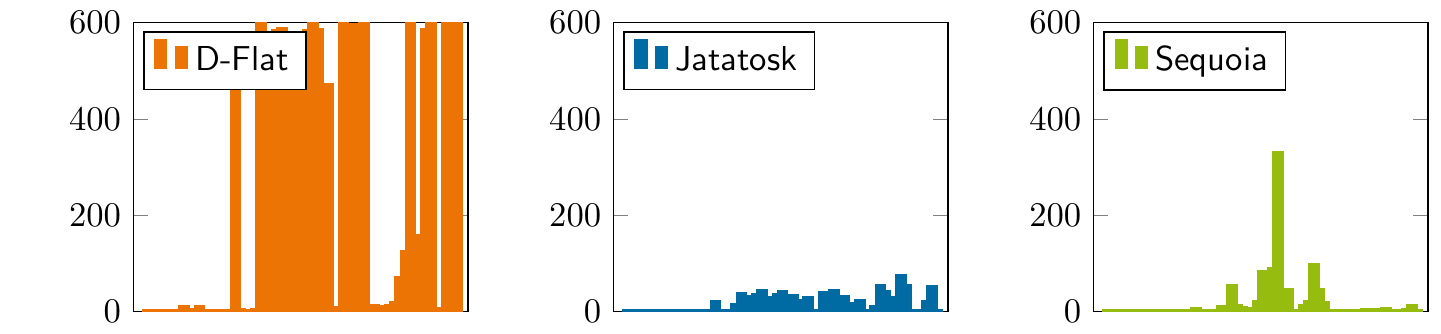}}%

  \subcaption{The left picture shows the \textcolor{col1}{difference}
    of \jatatosk\ against \dflat\ and \sequoia. A
    positive bar means that \jatatosk\ is faster by this amount in
    seconds, and a negative bar means that either \dflat\ or \sequoia\
    is faster by that amount. The bars are capped at $100$. On every instance, \jatatosk\ was compared
    against the solver that was faster on this particular instance. The
    image also shows for every instance the \textcolor{col2}{size} and
    the \textcolor{col3}{treewidth} of the input. The right image shows
    the number of instances that can be solved by each of the solvers in
    $x$ seconds, i.\,e., faster growing functions are better. The colors in this image
    are as in \textsf{(b)}.}
  \resizebox{\textwidth}{!}{\includegraphics{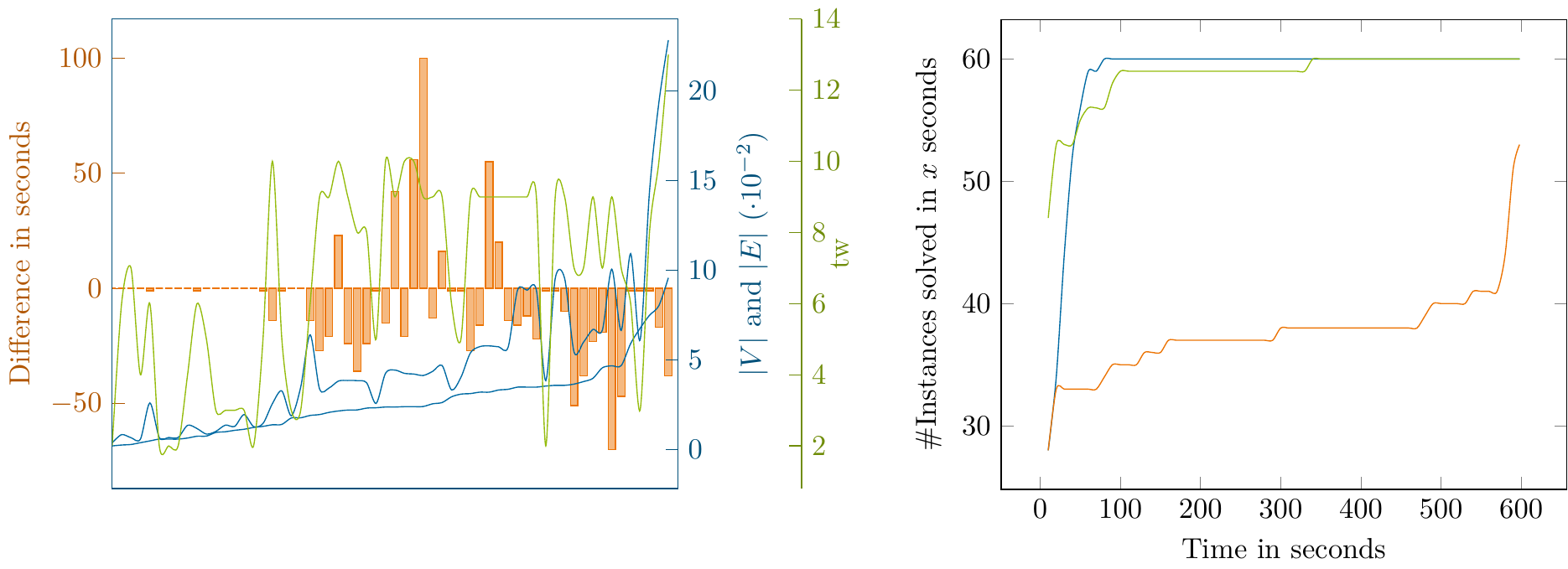}}%
\end{experiment}

  \begin{experiment}[h]
    \caption{$\Lang{feedback-vertex-set}$}
    \subcaption{Average, standard deviation, and median of the time (in seconds)
      each solver needed to solve $\Lang{feedback-vertex-set}$ over all
      instances of the data set from Section~\ref{section:Experiments}. The best
      values are highlighted. To solve the problem in \sequoia, we used the
      well-known fact that $F\subseteq V(G)$ is a feedback vertex set if for all
      non-empty subsets $W\subseteq V(G)\setminus F$, the induced graph $G[W]$
      contains a vertex of degree at most one.}
  \label{table:feedbackSet}
  \begin{tabular}{llll}
    \toprule
    & \dflat & \jatatosk & \sequoia\\
    \cmidrule(rl){1-4}
    Average Time & 587.84 & \textbf{303.6} & 384.16\\
    Standard Deviation & \textbf{77.72} & 292.76 & 282.28 \\
    Median Time & 600 & \textbf{548} & 600\\
    \bottomrule
  \end{tabular}
  \vspace{.5cm}

  \subcaption{Comparison of solvers for the \Lang{feedback-vertex-set} problem on
    the complete data set of Section~\ref{section:Experiments}. The
    time is measured in seconds and the instances are sorted by the
    number of vertices, see also \textsf{(c)}.
    \label{figure:feedbackSet}
  }
  \centering%
  \resizebox{\textwidth}{!}{\includegraphics{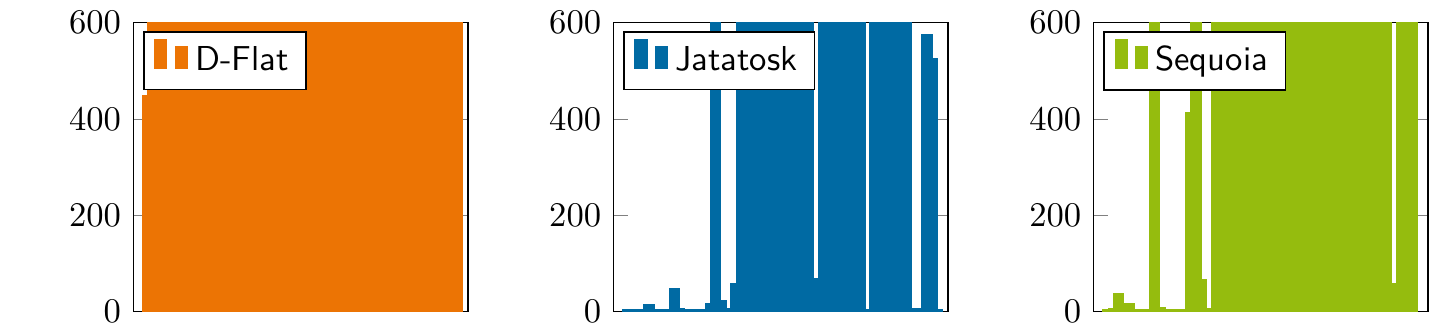}}%

  \subcaption{The left picture shows the \textcolor{col1}{difference}
    of \jatatosk\ against \dflat\ and \sequoia. A
    positive bar means that \jatatosk\ is faster by this amount in
    seconds, and a negative bar means that either \dflat\ or \sequoia\
    is faster by that amount. The bars are capped at $100$. On every instance, \jatatosk\ was compared
    against the solver that was faster on this particular instance. The
    image also shows for every instance the \textcolor{col2}{size} and
    the \textcolor{col3}{treewidth} of the input. The right image shows
    the number of instances that can be solved by each of the solvers in
    $x$ seconds, i.\,e., faster growing functions are better. The colors in this image
    are as in \textsf{(b)}.}
  \resizebox{\textwidth}{!}{\includegraphics{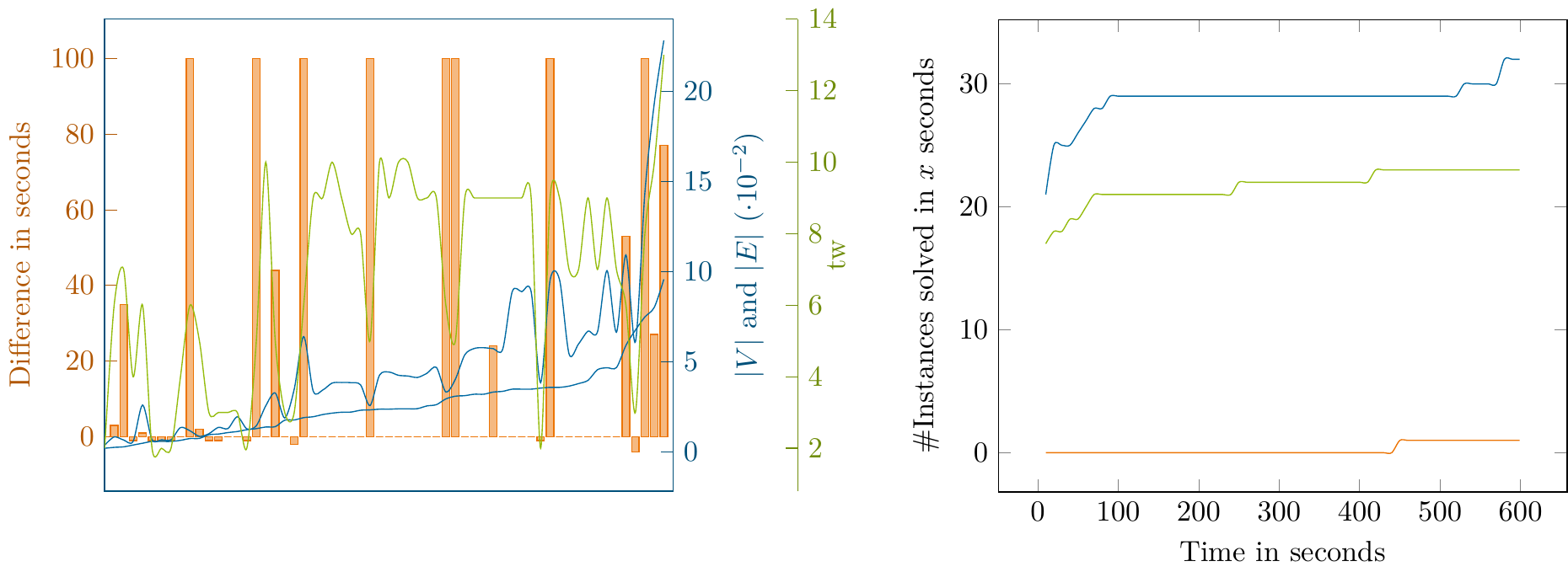}}
\end{experiment}
   
\end{document}